\documentclass[manuscript,nonacm]{acmart}
\pdfoutput=1
\usepackage{algorithm}
\usepackage{algorithmicx}
\usepackage{algpseudocode}
\usepackage{graphicx}
\usepackage{textcomp}
\usepackage{xcolor}
\usepackage{epstopdf}
\usepackage{natbib}
\usepackage{subfig}
\usepackage{multicol}
\usepackage{lipsum}
\usepackage{enumerate}
\usepackage{bm}

\newtheorem{claim}{Claim}

\begin{document}
\title[A New Fault-Tolerant Synchronization Scheme with Anonymous Pulses]{A New Fault-Tolerant Synchronization Scheme with Anonymous Pulses}

\author{Shaolin Yu}
\email{ysl8088@163.com}
\author{Jihong Zhu}
\author{Jiali Yang}
\author{Wei Lu}
\affiliation{
  \institution{Tsinghua University, Beijing, China}
}

\begin{abstract}
Robust pulse synchronization is fundamental in constructing reliable synchronous applications in wired and wireless distributed systems.
In wired systems, self-stabilizing Byzantine pulse synchronization aims for synchronizing fault-prone distributed components with arbitrary initial states in bounded-delay message-passing systems.
In wireless systems, fault-tolerant synchronization of pulse-coupled oscillators is also built for a similar goal but often works under specific system restrictions, such as low computation power, low message complexity, and anonymous physical pulses whose senders cannot be identified by the receivers.
These restrictions often prevent us from constructing high-reliable wireless synchronous applications.

This paper tries to break barriers between bounded-delay message-passing systems and classical pulse-coupled oscillators by introducing a new fault-tolerant synchronization scheme for the so-called anonymous bounded-delay pulsing systems in the presence of indeterministic communication delays and inconsistent faults.
For low computation power and low message complexity, instead of involving in consensus-based primitives, the proposed synchronization scheme integrates the so-called discrete mean-fields, planar random walks, and some additional easy operations in utilizing only sparsely generated anonymous pulses.
For fault-tolerance, we show that a square-root number of faulty oscillators can be tolerated by utilizing planar random walks in anonymous pulse synchronization.
For self-stabilization, we show that the stabilization can be reached in an expected constant number of observing windows in anonymous bounded-delay pulsing systems with the pulsing-frequency restriction.

Besides, the proposed synchronization scheme can also work in wired systems to eliminate the pulsing-frequency restriction.
For realization, the prototype experiment shows that this synchronization scheme can be easily implemented in real-world systems.
\end{abstract}

%
%


\maketitle

\section{Introduction}
\label{sec:Introduction}
Pulse synchronization is a basic building-block in constructing synchronous applications with real-world communication networks.
For example, the synchronization of pulse-coupled oscillators (PCO) in wireless sensor network (WSN) \citep{Hong2005scalable} is often required to be established with only sparsely generated physical pulses in saving overall energy consumption.
The synchronization schemes taken in VLSI systems-on-chips \citep{Fugger2006,Dolev2014PulseGeneration} are often built upon \emph{binary signals} (or saying the \emph{zero-bit messages}) or at most two-bits messages in avoiding \emph{too costly hardware design}.
Even the fault-tolerant time-triggered (TT) communication \citep{Steiner2008StartupRecovery} prefers \emph{semantic-free messages} or even some \emph{noises} (as some logical pulses) in reducing extra bandwidth consumption for synchronization.
In real-world systems, whatever the referred physical or logical pulses are, the core problem is establishing global synchronization between the distributed time-aware components via unit pulse-like interactions with non-ignorable communication delays.
Moreover, as real-world components are often unreliable, fault tolerance is desired in tolerating a number of faulty components.

For synchronizing distributed systems with real-world communication networks and unreliable components, several approaches have been proposed.
On the one hand, by assuming the faulty oscillators are fail-silent or fail-consistent, synchronization of PCO can be established with anonymous pulses in constant time \citep{Wang2020PCOByzantine}.
However, as slightly-out-of-specification (SOS) failures \citep{Ademaj03evaluationof} are common in many digital communication systems (an example in PCO is manifested as the detection problem in \cite{Hong2005scalable}), the reliability of synchronization systems with such a restricted fault assumption is gravely restricted even in just maintaining the synchronized state of the systems.
On the other hand, by assuming the faulty components can fail arbitrarily, i.e., being Byzantine, robust self-stabilizing synchronization can be established with liner or expected logarithmic time in the deterministic or randomized system settings, respectively \citep{Lenzen2019AlmostConsensus}.
However, although the probabilities of uncovered component failures can be eliminated in the Byzantine-fault-tolerant (BFT) synchronization schemes, the BFT solutions often lack sufficient scalability in computation, storage, and messages.
The gap is that, on the one side, as the complexity of BFT protocols is typically high, further optimization of the existing self-stabilizing BFT synchronization schemes is not an easy task in establishing the initial synchronization with arbitrary initial states then maintaining the synchronized states of the system deterministically.
on the other side, existing easier and cheaper PCO synchronization solutions consider only \emph{benign} faults \citep{Hong2005scalable}, \emph{consistent} faults \citep{Wang2018Attack, Wang2018Stealthy,Wang2020PCOByzantine}, or even constant communication delays, and thus may lack sufficient assumption coverage \citep{Powell1992assumption} in coping with \emph{malign} and \emph{inconsistent} faults in real-world scenarios.

In this paper, to make a better tradeoff between cost, efficiency, fault-tolerance, and timely startup of the pulse-coupled synchronization systems, we propose a new fault-tolerant synchronization scheme by introducing the so-called \emph{discrete mean-field} (DMF) with \emph{planar random walk} (PRW) \citep{Rayleigh1880,Pearson1905Random} in the anonymous bounded-delay (ABD) pulsing systems.
Firstly, by integrating traditional PCO networks with the bounded-delay networks \citep{DolevPulseBoundedDelay2007}, the ABD pulsing system is an extension of the PCO networks by allowing bounded and inconsistent communication delays of the unit pulses.
With the ABD pulsing systems, the DMF is an extension of the continuous mean-field investigated in complex networks with continuous interactions \cite{Dorfler2014Synchronizationcomplexsurvey}.
In digital networks, as the interactions between the oscillators are pulse-like, the DMF is defined with sparsely generated unit pulses other than the sinusoidal signals handled in \cite{Dorfler2014Synchronizationcomplexsurvey}.
Then, with the DMF, the relationship between the classical PRW and randomized fault-tolerant synchronization would be investigated in the ABD pulsing systems.
With this, several DMF-PRW synchronization primitives would be proposed, with which a self-stabilizing fault-tolerant synchronization solution for the ABD pulsing systems would also be included.

The benefit of this work is as follows.
Firstly, with the proposed DMF-PRW synchronization scheme, a shortcut between the classical random walks and the randomized fault-tolerant synchronization is directly built without involving high-complexity operations.
With the basic properties of the classical PRW in the $n$-oscillators pulsing system, the provided DMF-PRW synchronization primitives can work in the presence of $O(\sqrt{n})$ Byzantine faulty components.
More faulty components can be tolerated with specific \emph{malignities} without any changing of the proposed synchronization scheme.
Secondly, as the anonymous DMF can be utilized as a global timing reference without differentiating the senders, practical fault-tolerant synchronization can be more easily realized without costly hardware design.
Thirdly, by analyzing the basic self-stabilizing fault-tolerant synchronization solution, the core self-stabilization problem is expressed with the \emph{fixed-length curves moving} game.
With this, faster self-stabilizing BFT synchronization of the ABD pulsing systems can be further explored with both practical and theoretical interests.

The remainder of this paper is organized as follows.
In section~\ref{sec:Related Works}, the related works are briefly reviewed.
In section~\ref{sec:Models}, the basic system settings and the synchronization problem are given.
In section~\ref{sec:DMF}, the DMF and its basic properties are investigated.
In section~\ref{sec:Solutions}, we discuss how to build synchronization systems with the DMF by integrating the random walks with the mirror, revising, and other possible operations.
In section~\ref{sec:Realizations}, we give a prototype realization to show the basic applicability of the DMF-based synchronization scheme.
Lastly, the paper is concluded in section~\ref{sec:Conclusion}.

\section{Related Works}
\label{sec:Related Works}

Pioneering studies \cite{Huygens1665, Buck1966Biology, Buck1968Mechanism, Winfree1967Biological, Friesen1975Physiological, FriesenSynaptic, Peskin1975Mathematical, ArthurT1980The, Winfree1987Timing, Kuramoto1984, MSBiological1990} on coupled oscillators draw great attention of people in different disciplines.
However, it was not until recent years, progresses in the studies of the coupled oscillators, especially the ones with episodic and pulse-like interactions, are utilized to implement synchronization systems with digital communication networks \cite{Hong2003synchronization, DaliotBiological2003, Hong2005scalable, Degesys2007DESYNC, Patel2007Desynchronization, Degesys2008Desynchronization, Degesys2008DesynchronizationTopologies, Kang2009LocalizedDesynchronization, Pagliari2011Scalable, Nunez2012Bio, Wang2012Energy, Ashkiani2012DiscreteDesynchronization, Hinterhofer2012RD2,Wang2012OptimalPRF, Mauroy2012Kick, Wang2013Increasing, Nunez2015Global, Nunez2015Synchronization, Ferrante2016hybrid, Ferrante2017Robust, Gao2017Desynchronization, Wang2018Attack, Wang2018Stealthy,Wang2020PCOByzantine}.

Inspired by \cite{MSBiological1990}, an early scalable PCO synchronization scheme for WSN \cite{Hong2005scalable} is proposed with numerical simulations.
It is shown that by adopting the classical weakly-coupling strategy, global synchronization can be reached in delayed PCO networks with a sufficient number of oscillators \citep{Hong2005scalable}.
Meanwhile, by setting a relatively high detection threshold in receiving the pulses to avoid false detections, the startup time (or saying \emph{locking time} \citep{Hong2005scalable}) of the synchronization system can be in the order of several pulsing periods in the presence of missed detection of pulses.
However, as the simulations are performed with \emph{averaged network realizations}, the result cannot cover all worst-case scenarios.
Especially, as it is known that the basic coupling mechanism utilized in \cite{MSBiological1990} would fail to synchronize the PCO system in a finite time with some particular initial states, the simulations fail to cover such cases even without considering any faults of the oscillators.
Later works in this approach \citep{Pagliari2011Scalable, Wang2012OptimalPRF, Nunez2012Bio, Wang2012Energy, Wang2013Increasing, Nunez2015Global, Nunez2015Synchronization, Gao2017Desynchronization} show that real-world systems with short startup time can be built with large coupling strengths and optimized phase response functions.
However, the necessary fault tolerance is not covered in the provided solutions and prototype experiments.
In \cite{Wang2018Attack, Wang2018Stealthy,Wang2020PCOByzantine}, some restricted fail-consistent faults are handled in PCO systems.
These restricted fail-consistent faults are referred to as the \emph{stealthy Byzantine attacks} \citep{Wang2018Attack}, since the classical PCO model assumes reliable broadcast and the babbling-idiot faults can often be easily detected (and thus be isolated or removed from the systems).
With reliable broadcast, analysis shows that global synchronization can be reached with arbitrary initial states without considering communication delays \citep{Wang2018Stealthy,Wang2020PCOByzantine}.
However, as is investigated in \cite{Ademaj03evaluationof,Hong2005scalable}, inconsistent failures and communication delays of broadcast may be non-ignorable in real-world communication systems.
In building reliable synchronization in real-world PCO systems, the synchronization protocols should better tolerate inconsistent failures of the faulty oscillators in the communication networks with non-ignorable delays.
This is also the basic system setting originally proposed in \cite{Hong2005scalable}.

Inspired by \cite{Friesen1975Physiological, FriesenSynaptic}, an early deterministic self-stabilizing BFT pulse synchronization solution \citep{DaliotBiological2003} is proposed with bounded-delay networks.
It is shown that deterministic pulse synchronization can be reached in completely connected bounded-delay networks with arbitrary initial states in the presence of a linear number of Byzantine faults.
In the bounded-delay model, the Byzantine faults cover all possible faults of the distributed components and thus guarantee the highest component fault assumption coverage.
Lower-bounded by classical results \citep{Dolev1986Possibility,Dolev1982StrikeAgain}, the resilience of the BFT pulse synchronization solution in \cite{DaliotBiological2003} is optimal as the system can tolerate $f$ Byzantine nodes in the overall $3f+1$ nodes.
Later works in this approach show that deterministic linear stabilization time (the startup time measured with respect to the maximal message delay) can be reached without reliable broadcast.
However, the message complexity of deterministic self-stabilizing BFT pulse synchronization is high.
Meanwhile, although the complexity of randomized self-stabilizing BFT pulse synchronization \citep{Dolev2014PulseGeneration} can be much lower, the expected stabilization time (or saying the average startup time) is linear to the number of all distributed components in the system even with authenticated pulses whose senders can be deduced in the receivers.
In further reducing the startup time, a common-coin-integrated \citep{FM1989,FM1997} randomized solution proposed in \cite{Lenzen2019AlmostConsensus} can reach expected logarithmic startup time.
However, as the common-coin protocols proposed in \cite{FM1989,FM1997} require high polynomial message complexity, the overall complexity of the expected sub-linear time synchronization solution is high \citep{Lenzen2019AlmostConsensus}.
On the whole, although the optimal-resilient self-stabilizing BFT pulse synchronization problem is almost as easy as consensus \citep{Lenzen2019AlmostConsensus}, the integration of exact Byzantine agreement and pulse synchronization is a bit convoluted in the state-of-the-art solutions.
Despite the optimal resilience, it is also interesting to ask if self-stabilizing BFT pulse synchronization can be reached more naturally and straightforwardly.

Other approaches also exist in synchronizing the pulse-coupled systems to perform various synchronous services.
By noticing that some synchronous services such as TDMA communication can be performed with evenly distributed pulsing phases of the oscillators rather than the consistent pulsing phases, synchronization is established with the desired \emph{desynchronization} \citep{Degesys2007DESYNC, Patel2007Desynchronization, Degesys2008DesynchronizationTopologies, Kang2009LocalizedDesynchronization, Ashkiani2012DiscreteDesynchronization, Hinterhofer2012RD2, Gentz2016PulseSS}.
Namely, in the desynchronization approach, the oscillators would directly form some desired pulsing patterns in coordinating the distributed components.
As desynchronization can often be reached with localized synchronization operations, the complexity and networking requirement can be lower than direct global synchronization.
However, without a consistent global reference, the classical fault-tolerant strategies cannot be directly employed in the desynchronization systems in the presence of malign faults.
This restricts the application of desynchronization strategies in reliable systems.
Also, the efficiency of the ad-hoc schedules in the desynchronization solutions is limited compared to the synthetically optimized TT schedules.
Despite the desynchronization approach, there are also other coupling strategies for synchronizing PCO, such as inhibitory coupling \cite{Degesys2007DESYNC,Pagliari2010Bio} with elaborated refractory periods \cite{Konishi2008Synchronizationwithrefractory,Okuda2011Experimental,Buranapanichkit2015Desynchronization}.
However, with our limited knowledge, no PCO synchronization solution considers both inconsistent faults and non-ignorable communication delays with the assumption of a malicious adversary.

\section{System model and the problem}
\label{sec:Models}
The discussed ABD pulsing system $\mathcal{S}$ consists of $n$ nodes, denoted as $V$ with $|V|=n$, that are completely connected in the bidirectional communication network $G=(V,E)$.
With $G$, every node in $V$ can generate and send pulses to all nodes in $V$.
In considering fault-tolerance, the faulty nodes, denoted as $F$ with $|F|\leqslant f$, can make their pulses being inconsistently received in any portion of the nonfaulty nodes $Q=V\setminus F$.
Generally, the pulses can be anonymous or authenticated in the pulsing systems.
In this paper, we prefer the anonymous pulses to the authenticated ones.
In considering the anonymous pulses, we assume that all pulses are generated under the pulsing-frequency restriction.
Namely, over-frequent pulses are not allowed to be generated in any node of $\mathcal{S}$.
Obviously, in wired systems where the identities of the pulse senders can be deduced in the receivers, this pulsing-frequency restriction can be eliminated.

Following the traditional PCO system settings \citep{MSBiological1990,Hong2005scalable}, every nonfaulty node $q\in Q$ maintains a local pulsing phase $\phi_q\in[0,1]$ in running the synchronization protocols.
The pulsing phases can be realized with hardware clocks or special physical devices.
For simplicity, we assume the pulsing phase $\phi_q$ is realized in $q$ with a hardware clock $C_q$ that counts the basic ticks with a finite counter $c_q(t)$ valued in $[[c_{max}]]$, where $c_{max}$ is a sufficiently large integer in scheduling the timers, and $[[m]]=\{0,1,\dots,m-1\}$ is the set of the smallest $m$ nonnegative integers.
For convenience, we assume that the unit of the reference time equals the ideal pulsing cycle that is statically scheduled as $T$ ticks with the hardware clocks.
In considering the imperfectness of the hardware clocks, we assume that $c_q(t)/T$ progresses with speed arbitrarily distributed in $[1-\rho,1+\rho]$ with respect to the reference time $t$, where $\rho$ is the maximal drift-rate of the hardware clocks.
If no \emph{reset} nor \emph{adjustment} of $\phi_q$ happens, $\phi_q(t)$ would progress with the same speed of $c_q(t)/T$.

In each nonfaulty node $q\in Q$, a pulse would be generated at instant $t$ if and only if (iff) $\phi_q(t)$ reaches $1$.
When $\phi_q(t)$ reaches $1$, besides generating the pulse, $q$ would immediately \emph{reset} $\phi_q(t^+)$ as $0$ if no \emph{adjustment} of $\phi_q$ happens (we ignore the speed of $\phi_q$ during $[t,t^+]$).
Besides \emph{resetting} $\phi_q$, $q$ can also \emph{adjust} $\phi_q$ to any values in $[0,1]$ instantaneously at any instant $t$ according to the provided synchronization algorithms.
As the drift rates of the hardware clocks are bounded, all nodes in $Q$ can approximately have the pulsing cycle $1$ if no \emph{adjustment} happens.
With this, we assume that every node cannot generate more than $2(1+\rho)\tau$ pulses in any $\tau$ duration.

Meanwhile, the pulse generated in $q\in Q$ at the pulsing instant $t$ would be sent to all nodes in $V$.
Without loss of generality, we assume the maximal overall delay in sending and processing a pulse between two nonfaulty nodes is bounded within some constant $d>0$ in the bounded-delay networks \citep{DolevPulseBoundedDelay2007}.
Thus, the pulse generated in $q$ at the pulsing instant $t$ would be received and recorded with a hardware timestamp (the receiving tick) in all nonfaulty nodes during $[t,t+d]$.
Further, we assume the overlapped pulses generated in different nodes can be correctly counted by measuring the received signal power \citep{Hong2005scalable} or the number of queued messages.

In considering the malicious adversary, we assume that
1) the adversary can arbitrarily schedule the actual delay within $[0,d]$ for every pulse sent from node $q\in Q$ to node $q'\in Q$;
2) the adversary can arbitrarily schedule the speed of $\phi_q(t)$ for every node $q\in Q$ within $[1-\rho,1+\rho]$, and arbitrarily change this speed in this range with the progress of $t$;
3) all the faulty nodes are under the control of the adversary who can make the pulses of $F$ being arbitrarily received in $Q$ with a malignity index $\sigma\in[0,1]$ (being defined later);
and 4) the adversary knows all the provided algorithms and the past and current states of all nonfaulty nodes in every execution of $\mathcal{S}$.
The so-called malignity index is an extension of the Byzantine malignity, with which the nodes would be Byzantine with $\sigma=1$.

In the presence of such a malicious adversary, all nonfaulty nodes are expected to be globally synchronized.
Here, to discuss the synchronization strategies in a general way, instead of directly employing the pulsing phase $\phi_q$, we use another phase variable $\Phi_q(t)\in [0,1)$ in every node $q\in Q$ to represent the synchronization phase of $q$ at instant $t$.
Similar to the pulsing phase $\phi_q$, if no adjustment of $\Phi_q$ happens, $\Phi_q(t)$ would progress with the same speed of $c_q(t)/T$ and would be reset to $0$ when passing $1$.
Meanwhile, for synchronization, $\Phi_q(t)$ can be periodically adjusted according to the measured global reference.

For conveniently representing the difference of two normalized phases $\Phi,\Phi'\in [0,1]$, we define
\begin{eqnarray}
\label{eq_sync_distance}
\mathring{d}(\Phi,\Phi')=\min\{(\Phi-\Phi')\bmod 1, (\Phi'-\Phi)\bmod 1 \}
\end{eqnarray}
with $\bmod$ being the modular operation.
With this, $(\Pi,\Delta)$-synchronization is reached since $t_0$ iff for all $t\in [t_0,+\infty)$ and all $q,q_1,q_2\in Q$
\begin{eqnarray}
\label{eq_sync_precision}
\mathring{d}(\Phi_{q_1}(t),\Phi_{q_2}(t))\leqslant \Pi ~~~~~~~~~~~~~~\\
\label{eq_sync_accuracy}
\forall t'\in[t,t+1/2]:\mathring{d}(\Phi_{q}(t')-\Phi_{q}(t),t'-t)\leqslant \Delta
\end{eqnarray}
hold, where $\Pi$ and $\Delta$ are respectively called the normalized synchronisation precision and half-cycle synchronisation accuracy.
The half-cycle synchronization accuracy is concerned here as every nonfaulty node would at most adjust its synchronization phase once in every $1/2$ pulsing cycle when the system is synchronized.
So, with the adjustments of the synchronization phases of nonfaulty nodes being bounded in $\Pi$ in the synchronized state of $\mathcal{S}$, the half-cycle synchronization accuracy can also be represented as $\Delta=\rho /2+\Pi$.
In this case, the synchronization quality is mainly determined by $\Pi$.

Besides the synchronization quality, in measuring the efficiency of the synchronization solutions, it is desired that the generated pulses are anonymous and are sparsely generated, the computation performed in every node is easy, the startup time for reaching synchronization is short, and the maximal number of the tolerated malign faults is large.
In this paper, besides the provided DMF-PRW synchronization primitives, we also make the following basic claim.
\begin{claim}
\label{claim_result}
$(\Pi,\Delta)$-synchronization can be reached in $\mathcal{S}$ with arbitrary initial states in an expected constant number of observing windows given in Section \ref{subsec:Measuring}.
\end{claim}

\section{Discrete Mean-Fields}
\label{sec:DMF}
The traditional conception of mean-field is mainly investigated with continuous sinusoidal interactions \cite{Dorfler2014Synchronizationcomplexsurvey}.
Here we first define the DMF in the extended PCO networks with the anonymous pulses and the authenticated ones.
\subsection{Anonymous Discrete Mean-Fields}
\label{subsec:DMF:ADMF}
In classical PCO networks, pulses are often generated as simple physical impulses.
In this situation, as the identities of the senders can hardly be deduced in the receivers, these pulses are referred to as anonymous pulses.
For an anonymous pulse generated in some nonfaulty node at some $t\geqslant t_0$, the \emph{global phase} of this pulse with respect to $t_0$ can be represented as $\theta_{t_0}(t)=(t-t_0)\bmod 1$.
With this, by denoting $P$ as the set of all pulsing instants of the nonfaulty nodes, the (anonymous, as default) DMF in the time interval $[t_{\mathtt{a}},t_{\mathtt{b}}]$ can be generally defined as a complex number
\begin{equation}
\label{eq_dmf}
z(t_{\mathtt{a}},t_{\mathtt{b}})=\int_{t_{\mathtt{a}}}^{t_{\mathtt{b}}} \!e^{2\pi j(t-t_{\mathtt{a}})}\sum_{t'\in P}\delta(t-t') \mathrm {d}t
\end{equation}
where $\delta(x)$ is the dirac delta function satisfying $\int_{-\infty}^{+\infty} \!\delta(x)\mathrm {d}x=1$ and $\delta(x)=0$ for all $x\notin [-\epsilon,\epsilon]$ with a sufficiently small $\epsilon$.

Representing $z=\mathtt{r}(z)e^{2\pi j\psi(z)}$ with $\mathtt{r}(z)\geqslant 0$ being the strength of $z$ and $\psi(z) \in [0,1)$ being the normalized phase-angle of $z$, the DMF has the following basic property.
\begin{lemma}
\label{lemma_basic_dmf}
If $t_{\mathtt{a}}\leqslant t_{\mathtt{b}}< t_{\mathtt{c}}\leqslant t_{\mathtt{d}}$, $\mathtt{r}(z(t_{\mathtt{a}},t_{\mathtt{c}}))>0$, $\mathtt{r}(z(t_{\mathtt{b}},t_{\mathtt{c}}))>0$, and $\mathtt{r}(z(t_{\mathtt{b}},t_{\mathtt{d}}))>0$, then
\begin{equation}
\label{eq_dmf_r}
\frac{|\mathtt{r}(z(t_{\mathtt{a}},t_{\mathtt{c}}))-\mathtt{r}(z(t_{\mathtt{b}},t_{\mathtt{d}}))|}{\mathtt{r}(z(t_{\mathtt{b}},t_{\mathtt{c}}))}\leqslant \gamma(t_{\mathtt{a}},t_{\mathtt{b}},t_{\mathtt{c}},t_{\mathtt{d}})
\end{equation}
\begin{equation}
\label{eq_dmf_arg}
\sin 2\pi \mathring{d}(\psi(z(t_{\mathtt{a}},t_{\mathtt{c}})),\psi(z(t_{\mathtt{b}},t_{\mathtt{d}})))\leqslant \gamma(t_{\mathtt{a}},t_{\mathtt{b}},t_{\mathtt{c}},t_{\mathtt{d}})
\end{equation}
hold with
\begin{equation}
\label{eq_gamma}
\gamma(t_{\mathtt{a}},t_{\mathtt{b}},t_{\mathtt{c}},t_{\mathtt{d}})=\frac{\mathtt{r}(z(t_{\mathtt{a}},t_{\mathtt{b}}))+\mathtt{r}(z(t_{\mathtt{c}},t_{\mathtt{d}}))}{\mathtt{r}(z(t_{\mathtt{b}},t_{\mathtt{c}}))}
\end{equation}
\end{lemma}
\begin{proof}
With (\ref{eq_dmf}), by definition we have $z(t_{\mathtt{a}},t_{\mathtt{c}})=z(t_{\mathtt{a}},t_{\mathtt{b}})+ e^{2\pi j(t_{\mathtt{b}}-t_{\mathtt{a}})}z(t_{\mathtt{b}},t_{\mathtt{c}})$.
So we have $z(t_{\mathtt{a}},t_{\mathtt{c}})-e^{2\pi j(t_{\mathtt{b}}-t_{\mathtt{a}})}z(t_{\mathtt{b}},t_{\mathtt{d}})=z(t_{\mathtt{a}},t_{\mathtt{b}})-e^{2\pi j(t_{\mathtt{c}}-t_{\mathtt{a}})}z(t_{\mathtt{c}},t_{\mathtt{d}})$.
With this, as $|\mathtt{r}(z(t_{\mathtt{a}},t_{\mathtt{c}}))-\mathtt{r}(z(t_{\mathtt{b}},t_{\mathtt{d}}))|\leqslant \mathtt{r}(z(t_{\mathtt{a}},t_{\mathtt{c}})-e^{2\pi j(t_{\mathtt{b}}-t_{\mathtt{a}})}z(t_{\mathtt{b}},t_{\mathtt{d}}))$ and $\mathtt{r}(z(t_{\mathtt{a}},t_{\mathtt{b}})-e^{2\pi j(t_{\mathtt{c}}-t_{\mathtt{a}})}z(t_{\mathtt{c}},t_{\mathtt{d}}))\leqslant \mathtt{r}(z(t_{\mathtt{a}},t_{\mathtt{b}}))+\mathtt{r}(z(t_{\mathtt{c}},t_{\mathtt{d}}))$, we have (\ref{eq_dmf_r}).
With $\mathtt{r}(z(t_{\mathtt{a}},t_{\mathtt{c}}))>0$, $\mathtt{r}(z(t_{\mathtt{b}},t_{\mathtt{c}}))>0$, and $\mathtt{r}(z(t_{\mathtt{b}},t_{\mathtt{d}}))>0$, as $\sin 2\pi \mathring{d}(\psi(z(t_{\mathtt{a}},t_{\mathtt{c}})),\psi(z(t_{\mathtt{b}},t_{\mathtt{c}})))\leqslant \mathtt{r}(z(t_{\mathtt{a}},t_{\mathtt{b}}))/\mathtt{r}(z(t_{\mathtt{b}},t_{\mathtt{c}}))$ and $\sin 2\pi \mathring{d}(\psi(z(t_{\mathtt{b}},t_{\mathtt{d}})),\psi(z(t_{\mathtt{b}},t_{\mathtt{c}})))\leqslant \mathtt{r}(z(t_{\mathtt{c}},t_{\mathtt{d}}))/\mathtt{r}(z(t_{\mathtt{b}},t_{\mathtt{c}}))$, we have (\ref{eq_dmf_arg}).
\end{proof}

\subsection{Authenticated Discrete Mean-Fields}
To show the difference between the anonymous pulses and the identified messages, here we also define the DMF with the so-called \emph{authenticated} pulses.
In classical message-passing systems, the senders of the messages can often be identified by the receivers.
For example, in wired communication networks, the identity of the immediate sender of a message can be deduced from the corresponding communication channel.
In general, the pulses with identified senders are referred to as the authenticated pulses.
With the authenticated pulses, denoting $p_{t_0}^{(q)}$ as the earliest pulsing instant of $q$ since $t_0$, the authenticated DMF with respect to $t_0$ can be defined as
\begin{equation}
\label{eq_admf}
z(t_0)=\sum_{q\in Q}e^{2\pi j(p_{t_0}^{(q)}-t_0)}
\end{equation}

Denoting $P_{t_0}^{(q)}$ as the set of all pulsing instants of the nonfaulty node $q$ since $t_0$, for all $t_1,t_2\in P_{t_0}^{(q)}$, we have
\begin{equation}
\label{eq_admf_d}
\mathring{d}(\theta_{t_0}(t_1), \theta_{t_0}(t_2))\leqslant \rho|t_1-t_2|
\end{equation}
if $q$ does not adjust its pulsing phase.
Thus, denoting
\begin{equation}
\varepsilon_\tau=\rho(\tau+(1-\rho)^{-1})
\end{equation}
the authenticated DMF has the following basic property.
\begin{lemma}
\label{lemma_basic_admf}
If no node in $Q$ adjusts its pulsing phase since $t_0$, for all $\tau\geqslant 0$ we have
\begin{equation}
\label{eq_admf_r}
\mathtt{r}(e^{2\pi j\tau}z(t_0+\tau)-z(t_0))\leqslant|Q|\varepsilon_\tau
\end{equation}
\end{lemma}
\begin{proof}
By definition, we have $e^{2\pi j\tau}z(t_0+\tau)-z(t_0)=e^{2\pi j\tau}\sum_{q\in Q}e^{2\pi j(p_{t_0+\tau}^{(q)}-t_0-\tau)}-\sum_{q\in Q}e^{2\pi j(p_{t_0}^{(q)}-t_0)}=e^{-2\pi jt_0}\sum_{q\in Q}(e^{2\pi j p_{t_0+\tau}^{(q)}}-e^{2\pi j p_{t_0}^{(q)}})$.
As $p_{t_0+\tau}^{(q)},p_{t_0}^{(q)}\in P_{t_0}^{(q)}$, the ideal pulsing cycle is $1$, and $q$ would always generate a pulse during $[t_0+\tau,t_0+\tau+(1-\rho)^{-1}]$, with (\ref{eq_admf_d}) we have $\mathtt{r}(e^{2\pi j p_{t_0+\tau}^{(q)}}-e^{2\pi j p_{t_0}^{(q)}})\leqslant \rho(\tau+(1-\rho)^{-1})$ and thus we have (\ref{eq_admf_r}).
\end{proof}

\subsection{Generating the Basic Discrete Mean-Fields}
Generally, with the hardware clock $C_q$, every node $q\in Q$ can schedule a pulsing timer $\kappa_q$ (would expire when it is caught up by $c_q(t)$) as
\begin{equation}
\label{eq_sch_p}
\kappa_q(t^+)=(c_q(t)+(1+ x_q(t))T)\bmod c_{max}
\end{equation}
when $\phi_q(t)$ reaches $1$ at $t$, where $ x_q(t)\in [-1/2,1/2]$ is the normalized adjustment of $\phi_q$ at $t$.
For example, if no adjustment of $\phi_q$ happens, $ x_q(t)=0$ always holds.
In randomly adjusting $\phi_q$, $ x_q(t)$ would be randomly distributed in $[-1/2,1/2]$.
In practice, $\phi_q$ can be adjusted at any instants other than just at the pulsing instants.
Here we just use $ x_q$ to simplify the core problem.

Assuming $ x_i\sim U(-1/2,1/2)$ for all $i\in [[N]]$ (with $U$ being the uniform distribution and $N$ being a sufficiently large integer), $ x_{i_1}$ and $ x_{i_2}$ are independent for every two different $i_1,i_2\in N$, and denoting the probability of the occurrence of $\mathtt{E}$ as $Pr(\mathtt{E})$, the basic result of random walks in the plane \citep{Rayleigh1880} can be restated in the complex plane as follows.
\begin{lemma}\citep{Rayleigh1880}
\label{lemma_random_walk_r}
For $r=O(\sqrt{N})$ we have $Pr(\mathtt{r}(\sum_{i\in [[N]]}e^{2\pi j  x_i})\geqslant r)\approx e^{-r^2/N}$.
\end{lemma}
\begin{proof}
See \cite{Rayleigh1880}.
\end{proof}

This basic result can also be extended with smaller $N$ and the overall distribution of $r$ \citep{Wan2010Random}.
For simplicity here, we assume $N$ is sufficiently large and $r=O(\sqrt{N})$.
With this, we ignore the errors in estimating the probabilities with Lemma~(\ref{lemma_random_walk_r}) and would approximately use $e^{-r^2/N}$ as the corresponding probabilities.

\subsection{Measuring Discrete Mean-Fields}
\label{subsec:Measuring}
Firstly, denoting the receiving ticks of all received pulses during the last $\omega T$ ticks in $q\in Q$ as $R_\omega^{(q)}$, the anonymous DMF can be measured in $q$ at $t$ as
\begin{equation}
\label{eq_measuring_dmf}
\hat{z}_\omega^{(q)}(t)=\int_{c_q(t)-\omega T}^{c_q(t)} \!e^{2\pi j(\frac{c-c_q(t)}{T}+\omega)}\sum_{c'\in R_\omega^{(q)}}\delta(c-c') \mathrm {d}c
\end{equation}

For instantaneous receptions (in digital networks), (\ref{eq_measuring_dmf}) can be simplified as
\begin{equation}
\label{eq_measuring_dmf_simp}
\hat{z}_\omega^{(q)}(t)=\sum_{c'\in R_\omega^{(q)}}e^{2\pi j(\frac{c'-c_q(t)}{T}+\omega)}
\end{equation}

In measuring the authenticated DMF, as the identities of the pulse senders cannot be forged, the node $q\in Q$ can record one and only one latest received pulse for every node in $V$ in a sufficiently long observing window.
With this, if no pulse from $i\in V$ is received during this observing window in $q$, $q$ can take $i$ as a faulty node and use the pulse of $q$ as that of $i$ in estimations.
Assuming this observing window is also $\omega T$ ticks long, the authenticated DMF can be measured in $q$ as
\begin{equation}
\label{eq_measuring_dmf_auth}
\hat{z}_{\omega,V}^{(q)}(t)=\sum_{i\in V}e^{2\pi j(\frac{c_{\omega,i}^{(q)}-c_q(t)}{T}+\omega)}
\end{equation}
where $c_{\omega,i}^{(q)}$ is the latest receiving tick of the pulse from $i$ during the last $\omega T$ ticks in $q\in Q$.
For convenience, the $\omega T$-ticks observing window in a nonfaulty node is called the $\omega$-window.

\section{Synchronization with the Walks}
\label{sec:Solutions}
\subsection{Authenticated One-Kick Synchronization}
As a warm-up exercise, let us proceed first with the authenticated pulses.
In this case, a trivial synchronization strategy generates some randomized DMF and then utilizes it as the global reference.
Concretely, during the startup of the system, every node $q\in Q$ would generate its first pulse with $ x_q\sim U(-1/2,1/2)$ and would not adjust $\phi_q$ afterward.
Assuming all nonfaulty nodes have generated their first pulses during $[t_0-\tau,t_0]$ independently and denoting
\begin{equation}
r_{\alpha,\tau}=\sqrt{\alpha(n-f)}- n\varepsilon_\tau
\end{equation}
the basic \emph{one-kick} DMF has the following property, providing that $n$ is sufficiently large.

\begin{lemma}
\label{lemma_one_kick_r}
$\forall \alpha\geqslant 0:Pr(\mathtt{r}(z(t_0))\geqslant r_{\alpha,\tau})\geqslant e^{-\alpha}$.
\end{lemma}
\begin{proof}
As $n-f\leqslant |Q|\leqslant n$, with Lemma~\ref{lemma_random_walk_r}, we have $Pr(\mathtt{r}(z(t_0-\tau))\geqslant \sqrt{\alpha(n-f)})\geqslant e^{-\alpha}$.
And with Lemma~\ref{lemma_basic_admf}, we have $|\mathtt{r}(z(t_0))-\mathtt{r}(z(t_0-\tau))|\leqslant |\mathtt{r}(z(t_0)-z(t_0-\tau))|\leqslant n\varepsilon_\tau$ and thus the conclusion holds.
\end{proof}

So, in this case, there can be a sufficiently large probability that the initial DMF would have sufficient strength.
With this, if the DMF can be well-measured in the nonfaulty nodes, initial synchronization can be directly reached in constant time (with respect to the pulsing cycles) with a sufficiently large probability in the presence of a number of fail-arbitrarily nodes (without over-frequent pulses).
For utilizing the \emph{one-kick} DMF in synchronization, denoting
\begin{equation}
\label{eq_epsilon_0}
\epsilon_0=(1+\rho)n(\varepsilon_{\omega/(1-\rho)+d}+d)
\end{equation}
the following result can be directly derived with the result of \cite{Rayleigh1880}.

\begin{lemma}
\label{lemma_one_kick_measure_arg}
If $\beta f\leqslant r_{\alpha,\tau}-\epsilon_0$ with some $\beta>3$, then there is at least a probability $e^{-\alpha}$ that for all $\tau\geqslant\tau_2\geqslant\tau_1\geqslant 0$ and all $q_1,q_2\in Q$
\begin{equation}
\label{eq_sync_0}
\sin 2\pi \mathring{d}(\psi(\hat{z}_{\omega,V}^{(q_1)}(t_0+\tau_1))-\psi(\hat{z}_{\omega,V}^{(q_2)}(t_0+\tau_2)),\tau_2-\tau_1)\leqslant \epsilon_1
\end{equation}
holds with $\epsilon_1=2(\epsilon_0/(f(\beta-1))+n\varepsilon_{\tau_2-\tau_1+\omega/(1-\rho)}$.
\end{lemma}
\begin{proof}
Firstly, with the proof of Lemma~\ref{lemma_basic_admf}, as the communication delays and clock drifts are all bounded in nonfaulty nodes, if there is no faulty node, for all $q\in Q$ we would have $\mathtt{r}(\hat{z}_{\omega,V}^{(q)}(t)-z(t_q))\leqslant  \epsilon_0$ and $\sin 2\pi  \mathring{d}(\psi(\hat{z}_{\omega,V}^{(q)}(t)),\psi(z(t_q))\leqslant \epsilon_0/\mathtt{r}(\hat{z}_{\omega,V}^{(q)}(t))$ with $t_q$ being the start instant of the latest $\omega$-window in $q$, i.e., $(c_q(t)-c_q(t_q))\bmod c_{max}=\omega$ with $t$ being the current instant.
Thus, with up to $f$ faulty nodes, we have $\mathtt{r}(\hat{z}_{\omega,V}^{(q)}(t))\geqslant \mathtt{r}(z(t_q))-\epsilon_0-f$.
By applying Lemma~\ref{lemma_one_kick_r}, there is at least a probability $e^{-\alpha}$ that $\mathtt{r}(\hat{z}_{\omega,V}^{(q)}(t))\geqslant r_{\alpha,t_{q}-t_0}-\epsilon_0-f$ holds.
In this case, as $\beta f\leqslant r_{\alpha,t_{q}-t_0}-\epsilon_0$, we have $\sin 2\pi  \mathring{d}(\psi(\hat{z}_{\omega,V}^{(q)}(t)),\psi(z(t_q))\leqslant \epsilon_0/(f(\beta-1))$.
Thus, with Lemma~\ref{lemma_basic_admf}, for all $q_1,q_2\in Q$ we have (\ref{eq_sync_0}).
\end{proof}

Then, by measuring the basic \emph{one-kick} DMF in every $q\in Q$, $q$ can adjust its synchronization phase $\Phi_q(t)$ with $\psi(\hat{z}_{\omega,V}^{(q)}(t))$ in tracking the phase of the measured DMF.
With Lemma~\ref{lemma_one_kick_measure_arg}, as the basic \emph{one-kick} DMF can be approximately measured in all nonfaulty nodes, the synchronization phases of the nonfaulty nodes can be directly synchronized with approximately the normalized precision $(\arcsin \epsilon_1)/(2\pi)$, providing that the initial DMF has sufficient strength with $f=O(\sqrt{n})$.
Nevertheless, as the drifts of the hardware clocks are controlled by the malicious adversary, the strength of the DMF can tend to $0$ in the long progress of $t$.
To prevent this, every node $q\in Q$ can generate some new DMF and periodically adjust the synchronization phases in the long progress of $t$.
However, some undesired initial states (after the generation of the first pulses) of the system can still prevent the nodes in $Q$ from reaching global synchronization.
The rest of the paper discusses how to handle this problem with anonymous pulses.

\subsection{Anonymous Random-Walk Synchronization}
Firstly, with anonymous pulses, the former measurement of DMF would not do well since the nodes may receive different numbers of pulses in their observing windows.
To compensate for the loss of the uniqueness of the authenticated pulses, instead of recording just one randomly generated pulse for each node, every node $q\in Q$ would record much more anonymous pulses for every node.

Concretely, the (anonymous) \emph{random-walk} DMF can be generated as follows.
Firstly, every node $q\in Q$ would generate its first pulse with $x_q\sim U(-1/2,1/2)$ during the startup of $q$ just as before.
Then, at every pulsing instant $t$, $q$ would decide the next $\kappa_q$ in (\ref{eq_sch_p}) with another independently sampled $x_q\sim U(-1/2,1/2)$.
Meanwhile, by recording the instants of the received pulses in $R_\omega^{(q)}$, the DMF is measured with (\ref{eq_measuring_dmf_simp}).
Denoting
\begin{equation}
\label{eq_r3}
r_{\alpha,\tau,\omega}=\sqrt{\alpha(n-f)\lfloor 2(\omega/(1+\rho)-\tau)(1-\rho)/3\rfloor}
\end{equation}
the \emph{random-walk} DMF has the following property.

\begin{lemma}
\label{lemma_long_walk_basic}
If $\min\{t_{\mathtt{c}}-t_{\mathtt{a}},t_{\mathtt{d}}-t_{\mathtt{b}}\}\geqslant\omega/(1+\rho)$ and $\max\{t_{\mathtt{b}}-t_{\mathtt{a}},t_{\mathtt{d}}-t_{\mathtt{c}}\}\leqslant\tau$ hold with $\omega\gg \tau$, then with at least a probability $e^{-\alpha}$ that
\begin{equation}
\label{eq_gamma_0}
\gamma(t_{\mathtt{a}},t_{\mathtt{b}},t_{\mathtt{c}},t_{\mathtt{d}})\leqslant\frac{4n(1+\rho)\tau+2n}{r_{\alpha,\tau,\omega}}
\end{equation}
\end{lemma}
\begin{proof}
With Lemma~\ref{lemma_random_walk_r}, as every node $q\in Q$ would at least generate $\lfloor 2(\omega/(1+\rho)-\tau)(1-\rho)/3\rfloor$ pulses during $[t_{\mathtt{b}},t_{\mathtt{c}}]$, $Pr(\mathtt{r}(z(t_{\mathtt{b}},t_{\mathtt{c}}))\geqslant r_{\alpha,\tau,\omega})\geqslant e^{-\alpha}$ holds.
As every node in $V$ cannot generate more than $2(1+\rho)\tau+1$ pulses in any $\tau$ duration, we have $\max\{\mathtt{r}(z(t_{\mathtt{a}},t_{\mathtt{b}}),\mathtt{r}(z(t_{\mathtt{c}},t_{\mathtt{d}}))\}\leqslant 2n(1+\rho)\tau+n$.
Thus, with (\ref{eq_gamma}), there is a probability $e^{-\alpha}$ that (\ref{eq_gamma_0}) holds.
\end{proof}

Denoting $\tau_0=3(1-\rho)^{-1}/2$, every node $q\in Q$ would at least generate one pulse in any $\tau_0$ duration.
For simplicity, assume that $q$ would measure the latest \emph{random-walk} DMF at every pulsing instant and use the phase-angle of the measured DMF to adjust its synchronization phase.
Then, for any time interval $[t_0,t_0+\tau_0]$, all nonfaulty nodes would adjust their synchronization phases with the measured \emph{random-walk} DMF at least once.
For measuring the \emph{random-walk} DMF, denoting
\begin{equation}
\label{eq_epsilonrsigma}
\epsilon_{\omega,\sigma}=(2(1+\rho)\omega(\epsilon_0+ f))^{\sigma}+2n(1+\rho)\tau_0+2\omega n d
\end{equation}
where $\sigma\in [0,1]$ is called the malignity index, the following result can be derived like that of the \emph{one-kick} one, providing that $\omega$ is sufficiently large and $\rho$ and $d$ are sufficiently small.
\begin{lemma}
\label{lemma_long_walk_measure_arg}
If all nonfaulty nodes have recorded the received pulses for at least $\omega T$ ticks at $t_0$ and $r_{\alpha,\tau_0,\omega}>\epsilon_{\omega,\sigma}$, then there is at least a probability $e^{-\alpha}$ that for all $t_1,t_2\in[t_0,t_0+\tau_0]$ ($t_1\leqslant t_2$) and all $q_1,q_2\in Q$
\begin{equation}
\label{eq_sync_1}
\sin 2\pi \mathring{d}(\psi(\hat{z}_\omega^{(q_1)}(t_1))-\psi(\hat{z}_\omega^{(q_2)}(t_2)),t_2-t_1)\leqslant \epsilon_2
\end{equation}
holds with $\epsilon_2= (2\epsilon_{\omega,\sigma}+2\rho(1-\rho)^{-1}\omega)/(r_{\alpha,\tau_0,\omega}-\epsilon_{\omega,\sigma})$.
\end{lemma}
\begin{proof}
Without loss of generality, assume $q_1$ and $q_2$ as respectively the first and the last nonfaulty nodes that measure the DMF during $[t_0,t_0+\tau_0]$,
With Lemma~\ref{lemma_long_walk_basic}, there is at least a probability $e^{-\alpha}$ $\gamma(t_1-\omega_{q_1},t_2-\omega_{q_2},t_1,t_2)\leqslant\frac{4n(1+\rho)\tau_0+2n}{r_{\alpha,\tau_0,\omega}}$ holds with $\omega_q$ being the actual duration of the current $\omega$-window in $q$.
In this case, similar to Lemma~\ref{lemma_one_kick_measure_arg}, denoting $z(t_2-\omega_{q_2},t_1)=z_0$, as a node can at most generate $2$ pulses in any $T$ ticks, $z_0$ would be measured as some $\hat{z}_q$ in every $q\in Q$ satisfying $\mathtt{r}(\hat{z}_q)\geqslant \mathtt{r}(z_0)-2(1+\rho)\omega(\epsilon_0+ f)$ and $\sin 2\pi \mathring{d}(\psi(\hat{z}_q),\psi(z_0)+\tau_q)\leqslant 2(1+\rho)\omega(\epsilon_0+ f)/\mathtt{r}(\hat{z}_q)$ with $\tau_q=(t_2-\omega_{q_2})-(t_q-\omega_q)$ and $t_q\in[t_0,t_0+\tau_0]$ being the measuring instant in $q$.
Then, just regarding $z(t_1-\omega_{q_1},t_2-\omega_{q_2})$ and $z(t_1,t_2)$ as adversarial pulses, we have $\sin 2\pi \mathring{d}(\psi(\hat{z}_\omega^{(q)}(t_q)),\psi(z_0)+\tau_q)\leqslant \frac{\epsilon_{\omega,\sigma}}{r_{\alpha,\tau_0,\omega}-\epsilon_{\omega,\sigma}}$ for all $q\in Q$.
So with at least a probability $e^{-\alpha}$ (\ref{eq_sync_1}) holds.
\end{proof}

Thus, with sufficiently large $\omega$ and sufficiently small $\rho$ and $d$, the \emph{random-walk} DMF can serve as the global reference with a probability $e^{-\alpha}$.
However, with this, the running system would randomly fail with high probabilities.
To avoid this, we further explore walks with some feedback loops.

\subsection{Anonymous Walks with the Basic Feedback Loop}
In performing anonymous walks with a simple feedback loop, every node in $Q$ can observe the strength of the current anonymous DMF and perform different operations with this strength being in different ranges.
For the simplest example, denoting the anonymous DMF measured in $q\in Q$ with the $\omega$-window as $\hat{z}_q(t)$ at instant $t$, if $\mathtt{r}(\hat{z}_q(t))<\mathtt{R_0}$ with some threshold $\mathtt{R_0}$, $q$ would perform the original random walk with $ x_q\sim U(-1/2,1/2)$ just as in generating the \emph{random-walk} DMF.
Otherwise, if $\mathtt{r}(\hat{z}_q(t))\geqslant\mathtt{R_0}$, $q$ would directly overwrite $\psi_q(t^+)$ as $\psi(\hat{z}_q(t))$.
On the whole, the DMF generated in this hybrid way is called the basic \emph{half-random-walk} DMF.

For the basic \emph{half-random-walk} DMF, by configuring a sufficiently large $\mathtt{R_0}$, the measurement errors of the DMF phase-angles would be bounded in a small range when the overwriting operation is performed (also holds for the mirror operation, the same below).
Meanwhile, with the basic result of Lemma~\ref{lemma_long_walk_measure_arg}, by configuring a sufficiently large $\omega$, the difference of the measured DMF in different nonfaulty nodes would also be bounded in a small range when the overwriting operation is performed in adjacent time.
So, denoting $t_k$ as the $k$th smallest pulsing instant (in the pulsing-instant set $P$ defined in Section~\ref{subsec:DMF:ADMF}), we can represent the core transition of the basic \emph{half-random-walk} DMF with the following discrete-time system equation
\begin{equation}
\label{eq_sys_eq}
\vec z(k+1)=\mathbf A(k)\vec z(k)+[1,0,\dots,0]'(\varepsilon(k)b(k)+u(k)(1-b(k)))
\end{equation}
where $\vec z(k)=[z(t_k),z(t_{k-1}),\dots,z(t_{k-N+1})]'$ is the normalized $N$-dimension DMF vector at $t_k$, $\varepsilon(k)$ is the normalized measurement error of the DMF in generating the $(k+1)$st pulse, $b(k)$ is a boolean value that determines if the overwriting operation would be performed, $u(k)$ is the unit (one step) random-walk input, and $\mathbf A(k)$ is the transition matrix in deriving the next DMF vector $\vec z(k+1)$ with the current one.
Concretely, the transition matrix $\mathbf A(k)$ can be expressed as
\begin{equation}
\label{eq_A_eq}
\mathbf A(k)=a(k)
\left[ \begin{matrix}
   {\vec 1'}  \\
   {\mathbf 0}  \\
\end{matrix} \right]+
\left[ \begin{matrix}
   {\vec 0'} & {0}  \\
   {\mathbf I} & {\vec 0}  \\
\end{matrix} \right]
\end{equation}
with $a(k)={b(k)}/{r(k)}$, $r(k)=\mathtt{r}(\sum_{i=k}^{k-N+1})$, and $b(k)$ being expressed as
\begin{eqnarray}
\label{eq_b_eq}
b(k)=\left\{
\begin{aligned}
1 &, & {r(k)\geqslant\mathtt{R_0}+\epsilon_\mathtt{max}}\\
0 &, & {r(k)<\mathtt{R_0}-\epsilon_\mathtt{max}}\\
0 \text{ or } 1&, & \text{for otherwise}
\end{aligned}
\right.
\end{eqnarray}
where $\epsilon_\mathtt{max}$ is the maximal measurement error of the DMF strength.
As we have $b(k)=0$ when $r(k)<\mathtt{R_0}-\epsilon_\mathtt{max}$, $r(k)$ can be safely represented as $\max\{\mathtt{R_0}-\epsilon_\mathtt{max},\mathtt{r}(\sum_{i=k}^{k-N+1}z(t_k))\}$ in computing $a(k)$.
Meanwhile, as the measured DMF can always be easily normalized to a unit complex number in every nonfaulty node, we can safely assume that $\varepsilon(k)$ would not change the strength of the measured normalized DMF $a(k)\sum_{i=k}^{k-N+1}z(t_k)$.
Thus, we have $a(k)\in\{0\}\cup[1/N,1/(\mathtt{R_0}-\epsilon_\mathtt{max})]$ with the bounded measurement error of the DMF strength.

In considering continuous-time transitions of the \emph{half-random-walk} DMF, although the numbers of the pluses being observed in the $\omega$-windows may be different, this kind of difference can be handled by approximate discretization with allowing an extra bounded measurement error of the DMF.
Meanwhile, in considering self-stabilization, although the initial records in the $\omega$-windows of the nonfaulty nodes can be arbitrarily configured by the malicious adversary, the records of pulses generated in the nonfaulty nodes would be from the real pulses generated in the nonfaulty nodes after any $\omega/(1-\rho)$ duration.
Since then, the records in the $\omega$-windows of all nonfaulty nodes at $t_{k+1}$ can be approximately represented as $\vec z(k)$ with bounded measurement errors.
Now with the sufficiently large $\mathtt{R_0}$ and $\omega$ and sufficiently small $\rho$ and $d$, we can count all these bounded errors into the overall measurement errors $\epsilon_\mathtt{max}$ and $\varepsilon(k)$.
So, the core self-stabilizing synchronization problem here is the stabilization of the simplified discrete-time system described in (\ref{eq_sys_eq}) with arbitrary initial DMF vector $\vec z(N)$.

To solve this problem, by denoting $\vec c(k)=[1,0,\dots,0]'(\varepsilon(k)b(k)+u(k)(1-b(k)))$ and $\mathbf A(k,i)=\mathbf A(k)\mathbf A(k-1)\cdots\mathbf A(i)$, for $k\geqslant 0$ we have
\begin{equation}
\label{eq_sys_express}
\vec z(k+N+1)=\mathbf A(k+N,N)\vec z(N)+\sum_{i=N+1}^{k+N}\mathbf A(k+N,i)\vec c(i-1)+\vec c(k)
\end{equation}
For filtering the measurement errors, as $a(k) \in\{0\}\cup[1/N,1/(\mathtt{R_0}-\epsilon_\mathtt{max})]$, the effect of $\varepsilon(k)$ on the system would be no worse than that with $a(k)\equiv 1/(\mathtt{R_0}-\epsilon_\mathtt{max})$.
With this, as the accumulated effect of the measurement errors in $z(t_{k+1})$ can be upper-bounded by $\sum_{i=0}^{N-1}(\mathtt{R_0}-\epsilon_\mathtt{max})^{-i}\varepsilon(k-i)$, we ignore the effect of $\varepsilon(k)$ in the core analysis.
With this and denoting $\bar b(k)=1-b(k)$ for convenience, we further simplify (\ref{eq_sys_eq}) as
\begin{equation}
\label{eq_sys_simp}
\vec z(k+1)=\mathbf A(k)\vec z(k)+[u(k)\bar b(k),0,\dots,0]'
\end{equation}

To gain an intuitive understanding of this, let us play a simple game of fixed-length curves moving on the plane.
Firstly, as all elements in $\vec z(k)$ are unit complex numbers, we draw a curve $\Gamma(0)$ on the plane with $z(t_i)$ ($i=1,2,\dots,N$) being the $i$th unit segment just like the planar random walks, where the first segment $z(t_1)$ and the last segment $z(t_N)$ are respectively referred to as the tail and head of $\Gamma(0)$.
Then, the curve $\Gamma(k)$ (with $k=0,1,2,\dots$) would move on the plane by iteratively growing a new head from the old one and destroying the old tail in maintaining its fixed length $N$.
In growing a new head from $\Gamma(k)$ to $\Gamma(k+1)$, the distance $R(k)=r(k+N)$ between the two endpoints of $\Gamma(k)$ (referred to as the tail point and the head point) would be referenced.
When $R(k)<\mathtt{R_0}-\epsilon_\mathtt{max}$, the new head would grow with a randomly generated unit complex number.
When $R(k)>\mathtt{R_0}+\epsilon_\mathtt{max}$, the new head would grow with the unit complex number parallel to the vector from the tail point to the head point.
Otherwise, when $\mathtt{R_0}-\epsilon_\mathtt{max}\leqslant R(k)\leqslant\mathtt{R_0}+\epsilon_\mathtt{max}$, the growing of the new head is controlled by the adversary in choosing one of the above two growing rules.
With this, the stabilization problem requires that the curve $\Gamma(k)$ should be an approximately straight line with a small $k$ from arbitrarily shaped $\Gamma(0)$.

For Claim~\ref{claim_result}, here we provide an informal analysis to sketch the essential stabilization property of this discrete-time system, which also indicates some close relationship between the so-called \emph{fixed-length curves moving} game and the proposed DMF-PRW synchronization systems.
Firstly, it is trivial to see that $\Gamma(k)$ would become an approximately straight line if $r(k)>\mathtt{R_0}+\epsilon_\mathtt{max}$.
Meanwhile, with Lemma~\ref{lemma_random_walk_r} and (\ref{eq_sys_simp}) it is easy to see $r(k)\geqslant\mathtt{R_0}-\epsilon_\mathtt{max}$ would be satisfied with a constant positive probability if $\forall k'\in\{k-N,\dots,k-1\}:r(k')<\mathtt{R_0}-\epsilon_\mathtt{max}$ and $\mathtt{R_0}=O(\sqrt{N})$ holds.
Also, if $b(k)\equiv 1$, with common eigenvalue analysis, $\Gamma(k)$ would converge into an approximately straight line with an exponential convergence rate.
So the main problem here is to show that the adversary cannot hold the strength of the DMF in $[0,\mathtt{R_0}+\epsilon_\mathtt{max}]$ with high probabilities by just manipulating $b(k)$ and $\Gamma(0)$.

In manipulating $\Gamma(0)$, the adversary might want to prevent the growth of the DMF strength by posing some convolutely shaped curve that changes its growing directions sharply and frequently.
However, with (\ref{eq_sys_express}) we see that the initial state $\vec z(N)$ would be iteratively transferred with the transition matrices $\mathbf A(k+N,N)$.
For every $\mathbf A(i)$ with $i\in\{N,\dots,k+N\}$, when $b(i)=1$, $\vec 1$ is an eigenvector corresponding to the maximal simple eigenvalue of $\mathbf A(i)$.
Meanwhile, in this case, as $a(k) \in [1/N,1/(\mathtt{R_0}-\epsilon_\mathtt{max})]$, the absolute values of all other eigenvalues of $\mathbf A(i)$ is below than $1$.
So, in this case, any $\vec z(N)$ would tend to be fast filtered into $c\vec 1$ with $c$ being some fixed complex number ($|c|\leqslant 1$) and thus the curve $\Gamma(k)$ would tend to be a line.
Otherwise, when $b(i)=0$, the contribution of $\vec z(N)$ to the new head of the curve is $0$ and thus $\Gamma(k)$ would still tend to be a line (or a dot with $c=0$ if $b(i)\equiv 0$) without changing its growing direction abruptly.
So the effect of $\Gamma(0)$ in changing the direction of the current moving curve would always be filtered out with at least an exponential convergence rate.
Thus, the system can be simplified as the following equation (with ignoring the trivial normalization factor) after a constant number of $\omega$-windows.
\begin{equation}
\label{eq_sys_express_1}
\vec z(k+N+1)=c\vec 1+\sum_{i=N+1}^{k+N}\mathbf A(k+N,i)\vec c(i-1)+\vec c(k)
\end{equation}

Meanwhile, the adversary can only manipulate $b(k)$ when $\mathtt{R_0}-\epsilon_\mathtt{max}\leqslant r(k)\leqslant\mathtt{R_0}+\epsilon_\mathtt{max}$ holds.
By replacing $\vec c(i)$ as $[u(i)\bar b(k),0,\dots,0]'$ in the right side of (\ref{eq_sys_express_1}), we see that the changes of the directions of the random walk inputs $u(i)$ also tend to be filtered out with $\mathbf A(k+N+1,i+1)$ with enough large DMF strength.
So, as the vector from the tail point to the head point of $\Gamma(k)$ is just the DMF $z(k+N)$, the expectation of the strength of $z(k+N)$ would not be lowered with the transition matrices $\mathbf A(k+N,N)$.
Concretely, with applying Lemma~\ref{lemma_random_walk_r}, there is at least a probability $\alpha$ that $\mathtt{r}(\sum_{i=k}^{k+N}u(i))\geqslant \sqrt{\alpha N}$ holds.
Now if the adversary can manipulate $b(k)$, we have $r(k)>\mathtt{R_0}-\epsilon_\mathtt{max}$.
By configuring $\mathtt{R_0}$ as $\beta_1\epsilon_\mathtt{max}$ with a sufficiently large $\beta_1$, there is a constant positive probability $\alpha_1$ that $\mathring{d}(\psi(z(t_k)),\psi(\sum_{i=k}^{k+N}u(i)))< \pi/4$ (or smaller if needed) holds, since the directions of the random walk inputs are uniformly distributed in all directions.
In this case, as the strength of $z(t_{i})$ would not decrease in performing the overwriting operation, the strength of $z(t_{k+N})$ would not be worse than that of random walks by manipulating $b(k)$, providing that the past states of the system (being destroyed as the tails of the moving curve) would not abruptly change the directions of $z(t_i)$ for more than $\pi/2$ with $i\in\{k+1,\dots,k+N\}$.
To show this, firstly, if $R(k)>\mathtt{R_0}-\epsilon_\mathtt{max}$ and the diameter (of the convex hull, the same below) of $\Gamma(k)$ is larger than $\mathtt{R_0}+\epsilon_\mathtt{max}+\sqrt{\alpha_0 N}$, then there is a constant probability $e^{-\alpha_0}$ that $R(k')>\mathtt{R_0}+\epsilon_\mathtt{max}$ with some $k'<k+N$.
So we only need to handle the case with the diameter of $\Gamma(k)$ being within $\mathtt{R_0}+\epsilon_\mathtt{max}+\sqrt{\alpha_0 N}$.
In this case, with a small $\alpha_0$, the upper-bound of the diameter of $\Gamma(k)$ is approximately $\mathtt{R_0}+\epsilon_\mathtt{max}$.
Similarly, the diameter of $\Gamma(k+N)$ can also be approximately upper-bounded by $\mathtt{R_0}+\epsilon_\mathtt{max}$ with the same consideration.
As $\mathring{d}(\psi(z(t_k)),\psi(\sum_{i=k}^{k+N}u(i)))< \pi/4$, although the convex hulls of $\Gamma(k)$ and $\Gamma(k+N)$ has a small intersection, the undesired overwriting operation can only happen with $\mathring{d}(\psi(z(t_i)),\psi(z(t_k)))>\pi/2$ and $\mathtt{r}(z(t_i))\geqslant \mathtt{R_0}-\epsilon_\mathtt{max}$.
Now with a sufficiently large $\mathtt{R_0}=\beta_1\epsilon_\mathtt{max}$, these two conditions cannot hold together.
So, the strength of $z(t_{k+N})$ would be larger than $\mathtt{R_0}+\epsilon_\mathtt{max}$ with a constant positive probability with $\mathtt{R_0}\leqslant \sqrt{\alpha N}$.

Then, the last problem here is to configure the related parameters.
On the one hand, we expect that the initial strength of the current DMF can reach $\mathtt{R_0}$ with a high probability.
For this, the threshold $\mathtt{R_0}$ should be configured sufficiently small.
However, on the other hand, in supporting a sufficiently strong initial DMF to tolerate faults and measurement errors, $\mathtt{R_0}$ is expected to be configured sufficiently large.
In reconciling this, the malignity index $\sigma\in[0,1]$ in (\ref{eq_epsilonrsigma}) can serve as a simple tradeoff.
Concretely, if $\sigma\in[0, 1/2]$, the faults in the system would be no worse than random walks, which can be referred to as the benign faults, as it means that the faulty nodes would at most generate an $O(\sqrt{\omega f})$-strength interference in $\epsilon_{\omega,\sigma}$.
In this case, $f=O(n)$ faulty nodes can be tolerated in the system with both satisfying $\mathtt{R_0}=\Theta(\epsilon_{\omega,\sigma})$ and $\mathtt{R_0}=\Theta(r_{\alpha,\tau_0,\omega})$.
For example, we can configure $\mathtt{R_0}=\sqrt{N}/2$ in both supporting an easy-reaching initial strength and sufficient resilience to faults and measurement errors.
If $\sigma\in(1/2,1)$, in satisfying $\epsilon_{\omega,\sigma}=\Theta((\omega f)^\sigma)=\Theta( r_{\alpha,\tau_0,\omega})=\Theta(\sqrt{n\omega})$, the system can at least tolerate $f=O(n^{(1-\sigma)/\sigma})$ faulty nodes with $\omega=O(n)$.
We can see that, by increasing $\sigma$, the strength of the faulty nodes grows stronger, and the system can only tolerate fewer such malign faults.
As $\sigma$ tends to $1$, the faulty nodes would tend to acquire the full strength of the Byzantine nodes without over-frequent pulses.
Then, when $\sigma=1$, more faulty nodes can be tolerated at the expense of a larger $\omega$.
However, a complete investigation of the bounds of the resilience and stabilization time is out of the range of this paper.
In further exploring these bounds, we note that the pulses in $z(t_1-\omega_{q_1},t_2-\omega_{q_2})$ and $z(t_1,t_2)$ are simply regarded as adversarial pulses in Lemma~\ref{lemma_long_walk_measure_arg}.
This can be better investigated by utilizing the pulses from nonfaulty nodes with the discussed discrete-time system.

\subsection{Integrating with other operations}

The basic \emph{half-random-walk} DMF can be further extended by integrating other optional operations.
For example, in making an intermediate gear between the original random walks and the deterministic overwriting operations, when $\mathtt{R_0}\leqslant\mathtt{r}(\hat{z}_q(t))<\mathtt{R_1}$ holds with some thresholds $\mathtt{R_0}$ and $\mathtt{R_1}$, $q$ can still first perform the original random walk with $ x_q\sim U(-1/2,1/2)$.
Meanwhile, $q$ would perform an additional mirror operation if the pulsing timer $\kappa_q$ scheduled in (\ref{eq_sch_p}) stays too far away from the phase angle of the measured DMF.
Concretely, denoting $\psi_q(t)=((\kappa_q(t)-c_q(t))\bmod c_{max})/T$, if the pulsing timer is scheduled as $\kappa_q(t^+)$ in the random walk and $\mathring{d}(\psi(\hat{z}_q(t)),\psi_q(t^+))>1/4$ holds with the current DMF being measured in $q$ as $\hat{z}_q(t)$ at the pulsing instant $t$, then $q$ would reschedule $\kappa_q$ as
\begin{equation}
\kappa_q(t^+)=(c_q(t^+)+\psi_q'(t^+)T)\bmod c_{max}
\end{equation}
with
\begin{equation}
\psi_q'(t^+)=(2\psi(\hat{z}_q(t))+\frac{1}{2}-\psi_q(t^+))\bmod 1
\end{equation}
Equivalently, this mirror operation can be viewed as adding a complex number $\hat{z}_q'=e^{2\pi j\psi_q'(t^+)}-e^{2\pi j\psi_q(t^+)}$ into the DMF when $q$ generates the pulse at the expiration of $\kappa_q$, as is shown in Fig.~\ref{fig:mirror}.
Meanwhile, the overwriting operation (would be performed when $\mathtt{r}(\hat{z}_q(t))\geqslant\mathtt{R_1}$) can be viewed as the composition of the original random walk operation, the mirror operation, and a revising operation which compensates the outcome of the mirror operation to the desired pulsing phase, as is shown in Fig.~\ref{fig:composition}.
Furthermore, the \emph{half-random-walk} DMF can be integrated with some variants of approximate agreement \citep{Dolev1986Approximate} on the unit circle (like the one employed in \cite{DolevS2004SSBFTClock}) when $\mathtt{r}(\hat{z}_q(t))\geqslant\mathtt{R_2}$ with some even larger threshold $\mathtt{R_2}$.
As is limited here, we mainly discuss the \emph{half-random-walk} DMF with the mirror operation.

\begin{figure}[htbp]
\centering
\subfloat[The mirror operation\label{fig:mirror}]{\centering\includegraphics[width=1.4in]{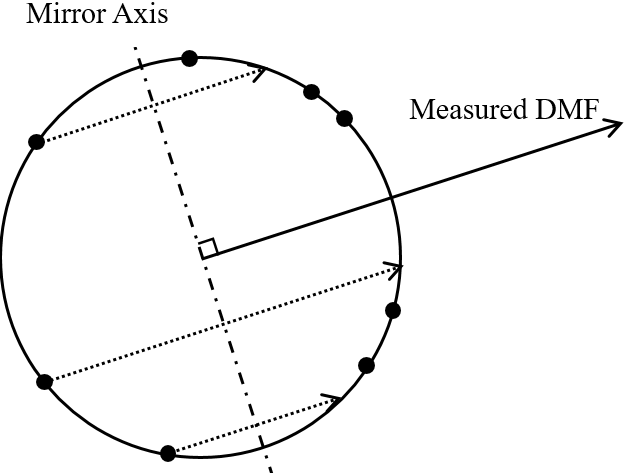}}
\subfloat[The composition\label{fig:composition}]{\centering\includegraphics[width=1.3in]{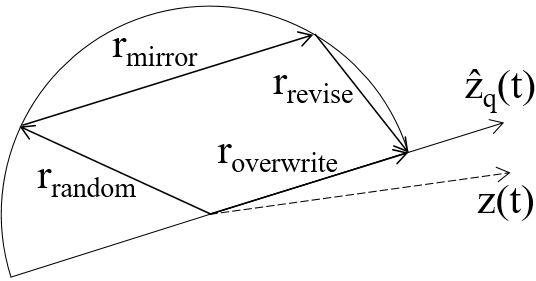}}
\subfloat[The mirror effect\label{fig:mirrorint}]{\centering\includegraphics[width=1.1in]{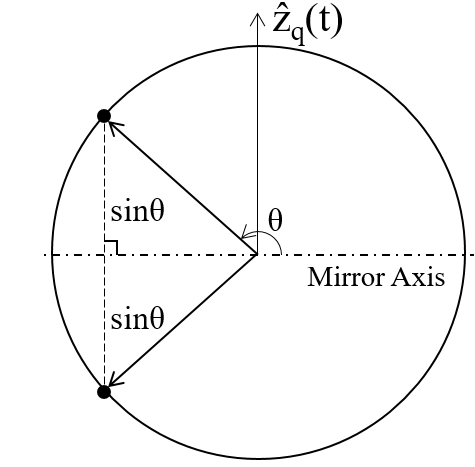}}
\caption{The extended half-random-walk DMF.}
\label{fig:carouselwalk}
\end{figure}

Intuitively, as the mirror operation is just adding an extra complex number being approximately parallel to the current DMF, it would generate much less effect to the original random walks in the orthogonal direction of the current DMF than that of the overwriting operation.
Meanwhile, if the strength of the measured DMF is always in $[\mathtt{R_0},\mathtt{R_1})$, as the samplings of the random variable $x_q$ are evenly distributed in $[-1/2,1/2]$, the expectation of the additional strength contributed by the mirror operation in a single step of the walk would be $\frac{1}{2}\int_{0}^{\pi}2\sin \theta\mathrm{d}\frac{\theta}{\pi}=\frac{2}{\pi}$, as is shown in Fig.~\ref{fig:mirrorint}.
With this, if the changes of the directions of the DMF are no larger than $\pi/2$ during the $N$-step walk, the strength of the DMF can grow to $\Theta(N)$ easily.

So, by configuring $\mathtt{R_1}=O(N)$ with some small linear coefficient, we expect that the system can reach a deterministic state with a high probability.
To this, we provide some simple numerical simulation results with randomly posed $\Gamma(0)$ without considering measurement errors and malign faults.
As the measurement errors are bounded, and the effect of $\Gamma(0)$ in changing the direction of the current moving curve would be fast filtered out, these results also make some sense in considering the general cases with arbitrarily posed $\Gamma(0)$.
Firstly, by configuring $\mathtt{R_0}=\sqrt{N}/2$ and $\mathtt{R_1}={N}/(2\pi)$ with $N$ being the number of the steps of the random walk observed in the $\omega$-window, the extended \emph{half-random-walk} DMF can reach the desired strength with very high probabilities.
As is shown in Fig.~\ref{fig:properR0R1}, by setting $N=100$ and running only $300$-step simulation (i.e., with $100$ records in every $\omega$-window and running every simulation for a $3\omega$ duration), fast stabilization is observed in all performed simulations (several $10000$ instances), where the DMF in all simulations are ended with the strength close to $N$.

\begin{figure}[htbp]
\centering
\subfloat[A typical walk\label{fig:goodR0R1}]{\centering\includegraphics[width=1.45in]{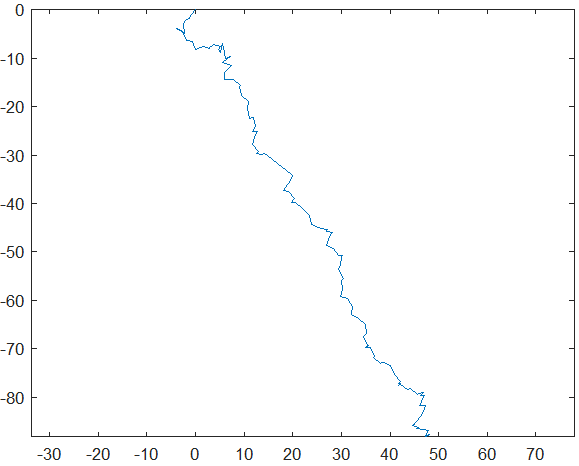}}
\subfloat[Distribution of strength\label{fig:goodR0R1sta}]{\centering\includegraphics[width=1.5in]{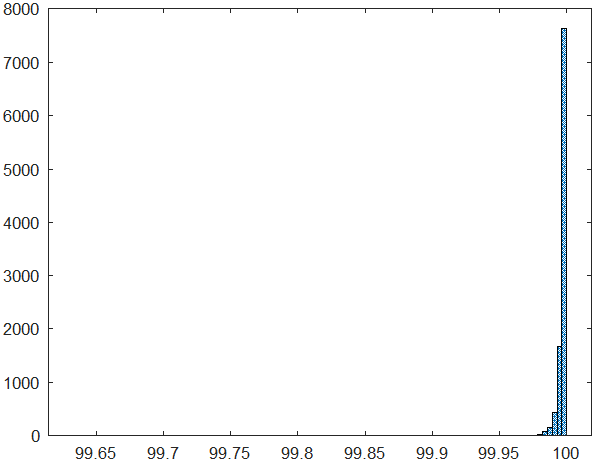}}
\caption{The walks with $\mathtt{R_0}=\sqrt{N}/2$ and $\mathtt{R_1}={N}/(2\pi)$.}
\label{fig:properR0R1}
\end{figure}

As a comparison, when the two parameters $\mathtt{R_0}$ and $\mathtt{R_1}$ are configured larger or smaller, the strength of the DMF might not always be fast converged to $N$ in $3$ $\omega$-windows.
For example, when $\mathtt{R_0}$ is configured as $2\sqrt{N}$ (the other settings are the same), the strength can be still in the $[0,\mathtt{R_0}]$ range after $3$ $\omega$-windows with a manifestable probability, as is shown in Fig.~\ref{fig:improperR0}.
For another example, when $\mathtt{R_1}$ is configured as $2{N}/\pi$, the walks show random behavior at the end of the simulations with a manifestable probability, as is shown in Fig.~\ref{fig:improperR1}.
Nevertheless, even with these larger and smaller parameters, the probabilities of system stabilization in $3$ $\omega$-windows are still very high.

\begin{figure}[htbp]
\centering
\subfloat[A typical walk\label{fig:largeR0}]{\centering\includegraphics[width=1.45in]{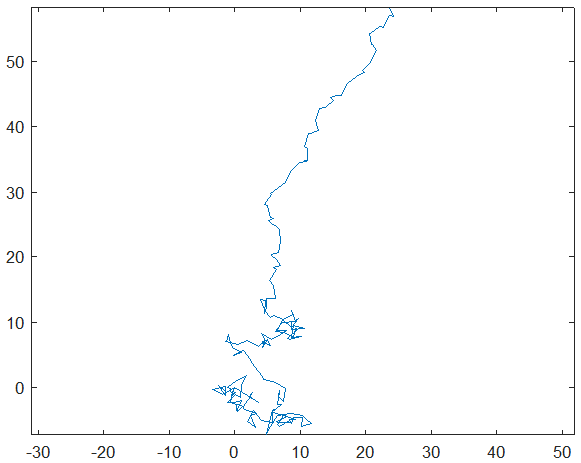}}
\subfloat[Distribution of strength\label{fig:largeR0sta}]{\centering\includegraphics[width=1.5in]{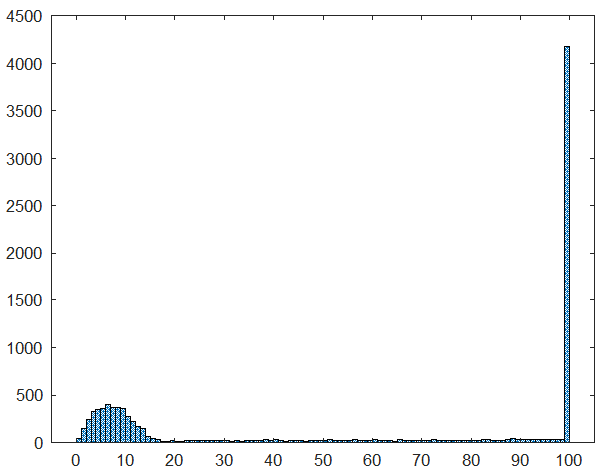}}
\caption{Walks with $\mathtt{R_0}=2\sqrt{N}$ and $\mathtt{R_1}={N}/(2\pi)$.}
\label{fig:improperR0}
\end{figure}

\begin{figure}[htbp]
\centering
\subfloat[A typical walk\label{fig:largeR1}]{\centering\includegraphics[width=1.5in]{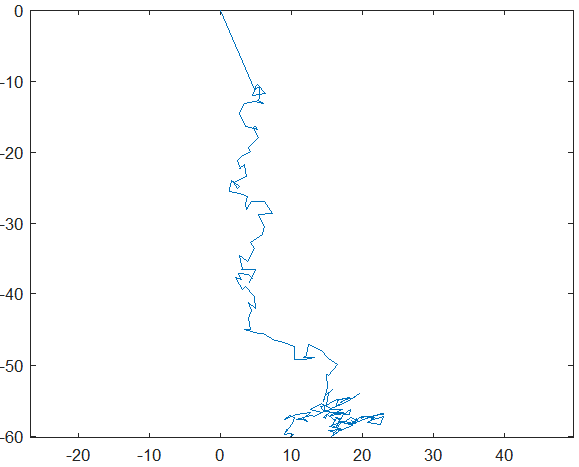}}
\subfloat[Distribution of strength\label{fig:largeR1sta}]{\centering\includegraphics[width=1.55in]{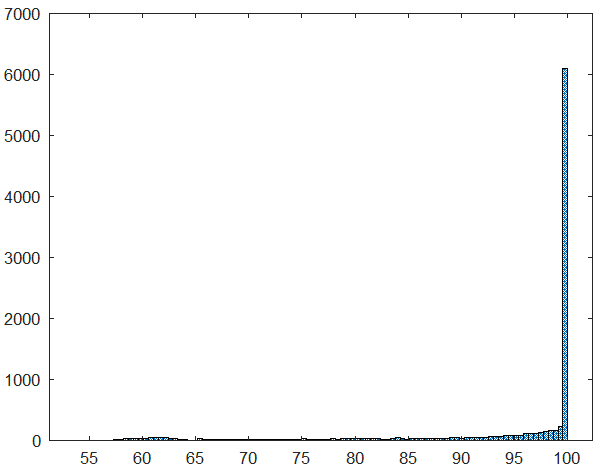}}
\caption{Walks with $\mathtt{R_0}=\sqrt{N}/2$ and $\mathtt{R_1}=2{N}/\pi$.}
\label{fig:improperR1}
\end{figure}

For the overall synchronization, with the strength of the DMF being more than $N/2$ with a very high probability after several $\omega$-windows, the precision $\Pi$ can be approximately $(\arcsin (2\epsilon/N))/(2\pi)$ with $\epsilon$ being the maximal errors in measuring the DMF.
So $\Pi$ can be improved by taking a larger $N$ (the number of items in the $\omega$-window) at the expense of a longer startup time (in running the first several $\omega$ durations).

\section{Realizations}
\label{sec:Realizations}
One practical problem with the basic DMF-based synchronization scheme might be the required trigonometric functions.
Here we show how to overcome this computation problem with a prototype system realization.

This basic realization is developed upon distributed computers connected with standard switched Ethernet to show the DMF-based synchronization's basic applicability.
Concretely, $14$ distributed nodes are configured as $14$ independent Linux kernel modules run in two computers with the Ubuntu16.04LTS \citep{ubuntu} operating systems.
Each node is scheduled to broadcast TT-frames at some pre-configured time-slots periodically.
To enhance the real-time properties of the system, we disable the power-save features of the CPUs, modify the NIC-driver codes in the kernel space, and make high-precision time readings.
Meanwhile, necessary atomic instructions and locks are also used to avoid resource access conflicts in each computer.
A full-duplex standard best-effort Ethernet switch connects the computers with a negotiated speed of 1000Mbps.

During initialization, each node $q\in Q$ is normally started with an empty queue of the recorded pulses.
Then, whenever the synchronization phase of $q$ reaches a pre-configured broadcasting phase (being different in different nodes), $q$ would broadcast a TT-frame.
Thus, on the one hand, the synchronization phases are expected to be synchronized by tracing the phase-angles of the DMF.
The DMF, on the other hand, is measured in every $q\in Q$ by just receiving the TT-frames from other nodes (by subtracting the pre-configured broadcasting phase of $i\in V$ from the receiving instant of the frame from $i$).

In realizing the DMF-based synchronization, as no floating-point numbers nor trigonometric functions are supported in the Linux kernel spaces, we have manually realized some lightweight trigonometric functions in the kernel.
These trigonometric functions only use integers to calculate and represent the input and output values, and thus can run much faster than the common trigonometric functions realized in glibc2.23 \citep{glibc}.
Then, in further simplifying these functions, we use the $1$-norm unit circle to replace the $2$-norm unit circle in simulating the sine and cosine functions.
Similarly, the $2$-argument arctangent functions can also be approximated by simple facets in the $3$-dimensional space.
Most simply, we use $8$ radially symmetrical facets, which make a continuous zigzag-surface and also have zero values when the phase-angle is $1/2$, to replace the original smooth surface of the $2$-argument arctangent functions.
More segments and facets can also be added to obtain lower approximation errors.

In verifying the basic applicability of the DMF-based synchronization, the experiments shown in Fig.~\ref{fig:realization} run with the simplest realization of the trigonometric functions and a trivial startup of the initial DMF (without any faulty nodes).
The frames are observed in a third computer connected to the switch to measure the synchronization results approximately.
In Fig.~\ref{fig:realization}, it shows that whenever the system is synchronized by the DMF, all the TT-frames can be received from the standard Ethernet switch with the pre-configured broadcasting phases.
Thus, we can see that the DMF-based synchronization can be easily realized in the presence of non-ignorable communication delays without common trigonometric functions.

\begin{figure}[htbp]
\centering
\subfloat[The synchronization process\label{fig:realization1}]{\centering\includegraphics[width=1.7in]{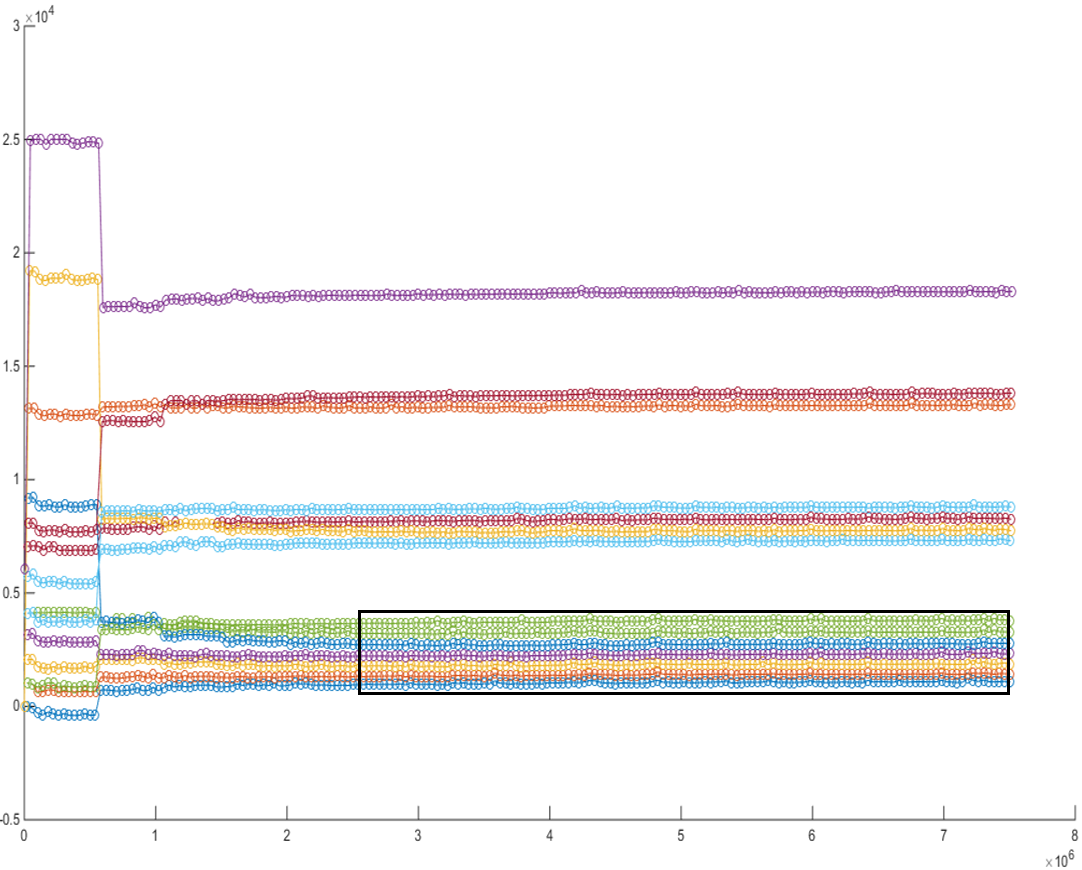}}
\subfloat[The stable state\label{fig:realization2}]{\centering\includegraphics[width=1.7in]{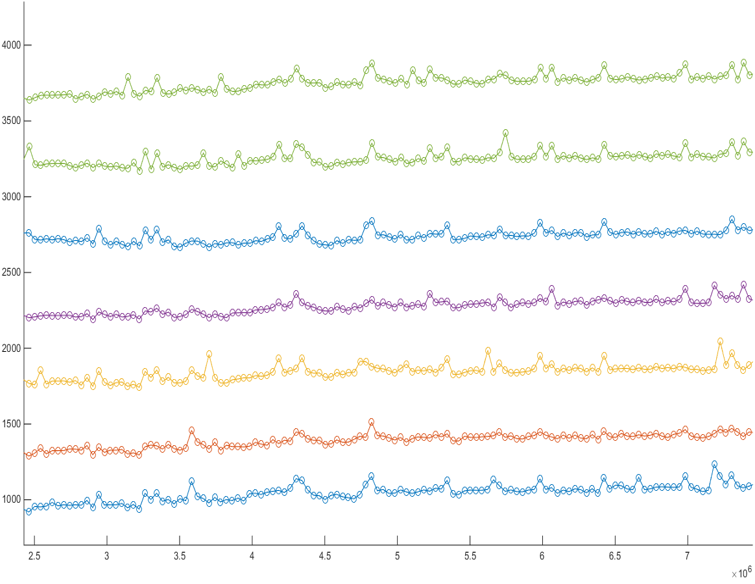}}
\caption{Real-world realization of DMF-based synchronization.}
\label{fig:realization}
\end{figure}

\section{Conclusion}
\label{sec:Conclusion}
In this paper, we have proposed the DMF-PRW synchronization scheme by integrating classical PRW with some additional operations.
Firstly, by introducing PRW into the ABD pulsing systems, we have proposed a way to utilize the classical result of PRW in establishing the initial synchronization with an expected sqrt-strength DMF.
For an intuitive example, we have shown that the authenticated one-kick DMF synchronization can trivially reach and maintain the synchronized state of the system for a bounded duration with some fixed probabilities.
Then, we have provided several kinds of anonymous DMF to construct more practical synchronization systems by introducing the mirror, revising, and other possible operations.
With a simplified analysis of the basic \emph{half-random-walk} system, we have shown that the effect of $\Gamma(0)$ in changing the direction of the DMF would be filtered out with an exponential convergence rate.
With this, we have also shown that the strength of the DMF can reach $\mathtt{R_0}+\epsilon_\mathtt{max}$ with a constant positive probability. Thus the synchronization system can reach stabilization in an expected constant number of $\omega$-windows.
Meanwhile, by introducing the mirror operation, we have shown that the strength of the DMF can fast grow to $O(N)$ with a high probability, with which the revising operation can further bring the DMF to the maximal strength.
With numeric simulations, we have shown that by properly configuring the core parameters, the \emph{half-random-walk} synchronization can be reached in a constant number of $\omega$-windows with very high probabilities.
Lastly, we have also provided a prototype realization and shown the basic applicability of the DMF-PRW synchronization without the support of common trigonometric functions.
As the pulsing-frequency restriction can be eliminated in wired systems, these results can be extended to wired message-passing systems without difficulties.

Despite the current results of the proposed system, several problems with the DMF-PRW synchronization scheme need to be further investigated.
Firstly, with a better investigation of the anonymous pulses generated in nonfaulty nodes, we wonder if the length of the $\omega$-windows can be reduced to a constant number being independent of $n$ and $f$, especially with $\sigma=1$.
Secondly, the analysis of the \emph{half-random-walk} synchronization system still relies on some informal arguments in integrating the classical eigenvalue analysis with planar random walks, especially in discussing the effects of past random walk inputs in changing the directions of the DMF.
This can be further improved with better investigation of the planar random walks.
Thirdly, the integration of the \emph{half-random-walk} system with other operations such as approximate agreements \citep{Dolev1986Approximate} can be further explored to enhance resilience in the stabilized systems.
In considering real-world applications, firstly, as the actual physical pulses are often generated with noises, the DMF-PRW synchronization solutions should better be extended to handle the general PRW with random-length steps \citep{wade2015convex,James2018planar}.
Secondly, as real-world communication networks are often partially connected, the stabilization properties of the DMF-PRW synchronization systems in partially connected networks need to be further studied.
Besides, by viewing the synchronized sub-networks as super-nodes in the larger networks, the scaling of system-sizes can also be further investigated with the DMF-PRW synchronization scheme.

\bibliographystyle{ACM-Reference-Format}
\bibliography{NFTSAP}

\end{document}